\documentclass[11pt,a4paper]{article}

\usepackage[T1]{fontenc}
\usepackage[utf8]{inputenc}
\usepackage{lmodern}
\usepackage[margin=1in]{geometry}
\usepackage{amsmath,amssymb,amsthm}
\usepackage{enumitem}
\usepackage{microtype}
\usepackage[hidelinks]{hyperref}

\numberwithin{equation}{section}
\newtheorem{theorem}{Theorem}[section]
\newtheorem{lemma}[theorem]{Lemma}
\allowdisplaybreaks[1]

\hypersetup{
  pdftitle={Exponential investors with weakly mean-reverting prices},
  pdfauthor={Bal\'azs Hoffmann and Mikl\'os R\'asonyi},
  pdfkeywords={mean-reversion, exponential utility, portfolio choice}
}

\title{Exponential investors with weakly mean-reverting prices\thanks{This research was supported by
grant K 143529 of the National Research, Development and Innovation Office (NKFIH).}}
\author{Bal\'azs Hoffmann
\quad and \quad
Mikl\'os R\'asonyi}
\date{}

\begin{document}
\maketitle
\thispagestyle{empty}

\begin{abstract}
We investigate a continuous-time financial market where the asset price exhibits weak (sublinear) mean reversion and has a nonzero drift. Complementing earlier work on strong (superlinear) mean reversion, we show that, for an investor maximizing expected exponential utility, the certainty equivalent grows as $O(T^{2\beta+1})$ where $0<\beta<1$ is the strength of mean reversion. An explicit asymptotically optimal strategy is also given.
\end{abstract}

\section{Introduction}

Mean-reversion has already been documented in prices of commodity futures (see \cite{ref1,ref3}) and they also arise in pairs trading procedures. Optimal investment for linear mean-reversion has been studied in detail, we refer to \cite{ref7}, \cite{ref2}.

In the present communication we continue the work initiated in \cite{ref6}: we study how the long-term asymptotic performance of investors with exponential utility is affected by the degree of non-linear mean-reversion.

Asset price is modelled by a mean-reverting process plus a linear drift term $\mu t$ for some $\mu\in\mathbb{R}$, here $t\geq0$ is the time variable. Results of \cite{ref6} have revealed that the cases $\mu=0$ and $\mu\neq0$ are qualitatively different. We stick to $\mu\neq0$ in the current paper, $\mu=0$ will be addressed elsewhere.

In \cite{ref6} the superlinear (``strongly mean-reverting'') regime has been explored (this corresponds to $\beta\geq1$ in the dynamics \eqref{eq:2.2} below). In that paper, the optimal growth rate of the certainty equivalent was determined as the trading horizon $T$ tends to $\infty$. It was found to be of the order $O(T^{2\beta+1})$ and a concrete asymptotically optimal strategy has been given.

We complement those results by investigating the sublinear case (that is, $0<\beta<1$), the case of ``weak mean reversion''. We establish the same growth rate $O(T^{2\beta+1})$ for the certainty equivalent and exhibit a concrete strategy with asymptotically optimal performance. Unlike in \cite{ref6}, the functions involved are not $C^2$ near the origin and thus they need to be smoothed out. Thus the asymptotically optimal strategy has a more complicated form in the present paper but it is analogous to the case $\beta\geq1$ away from the origin. The line of attack from \cite{ref6} fails in the present case, as explained below, and we take an alternative route for the proof.

\newpage
\section{Model description}
\label{sec:2}
We consider a financial market consisting of a riskless asset with
zero interest rate and a risky asset whose price process $S=(S_t)_{t\geq0}$ satisfies
\begin{align}
dS_t &= \mu\,dt+dX_t, \qquad S_0=0,\label{eq:2.1}\\
dX_t &= -\alpha\,\operatorname{sgn}(X_t)|X_t|^\beta\,dt+dB_t,
\qquad X_0=0,\label{eq:2.2}
\end{align}
where the latter process is mean-reverting around $0$. Here $0<\beta<1$, $\alpha>0$, $\mu\in\mathbb R$, and
$(B_t)_{t\geq0}$ is a standard Brownian motion on a filtered probability space
$(\Omega,\mathcal F,(\mathcal F_t)_{t\geq0},\mathbb P)$. We write
\begin{equation}\label{eq:2.3}
\theta_t:=\mu-\alpha\,\operatorname{sgn}(X_t)|X_t|^\beta,
\qquad\text{so that}\qquad dS_t=\theta_t\,dt+dB_t.
\end{equation}

The equation \eqref{eq:2.2} has a unique strong solution by e.g.\ Theorem 1.2 of \cite{ref8} since its drift
satisfies a one-sided Lipschitz condition.

For a finite horizon $T>0$, let $\mathcal A_T$ denote the class of predictable, $S$-integrable strategies
$H$ for which the gains process $V_t^H:=(H\cdot S)_t$, $t\in[0,T]$ is bounded from below by a constant.
Here $\cdot$ refers to stochastic integration.

An investor maximizing expected exponential utility with risk-aversion parameter equal
to one has optimal value
\begin{equation}\label{eq:2.4}
u_T:=\sup_{H\in\mathcal A_T}\mathbb E\bigl[-\exp\{-V_T^H\}\bigr],
\end{equation}
and the corresponding certainty equivalent is
\begin{equation}\label{eq:2.5}
c_T:=-\log(-u_T).
\end{equation}

If $\mathbb H=\{H^T\}_{T\geq1}$ is a family of strategies with $H^T\in\mathcal A_T$, we define its certainty-equivalent
order by
\begin{equation}\label{eq:2.6}
C(\mathbb H):=\sup\left\{\theta>0:\liminf_{T\to\infty}
\frac{-\log\mathbb E\exp\{-V_T^{H^T}\}}{T^\theta}>0\right\}.
\end{equation}
Our goal is to show that, for the model \eqref{eq:2.1} with $\mu\neq0$ no admissible family can have order
larger than $2\beta+1$, and that an explicit family of smoothed power strategies attains exactly
this order.

The paper is organized as follows. Section~3 states the main result. Section~4 proves
the upper bound. Section~5 constructs a smoothed power function, and section~6 derives
the terminal wealth identity and admissibility of the candidate strategy. Section~7 proves
exponential-moment estimate. The lower bound is proved in section~8 and in section~9.
Finally, section~10 completes the proof.

\section{Main results}
\label{sec:3}
Throughout the paper, $C$ with subscripts will refer to various constants
whose value we do not wish to make explicit at the given occurrence. Subscripts refer to the
quantities on which the given constant depends, for instance, $C_\mu$ is a constant that depends
(solely) on $\mu$.

\begin{theorem}[Optimal order of certainty equivalent growth for sublinear mean reversion]
\label{thm:3.1}
Assume $0<\beta<1$, $\alpha>0$, $\mu\neq0$. Then the following statements hold.
\begin{enumerate}[label=(\roman*)]
\item There is a constant $C_{\alpha,\beta,\mu}<\infty$ such that
\begin{equation}\label{eq:3.1}
c_T\leq C_{\alpha,\beta,\mu}T^{2\beta+1},\qquad T\geq1.
\end{equation}
Consequently every admissible family $\mathbb H=\{H^T\}_{T\geq1}$ satisfies
\begin{equation}\label{eq:3.2}
C(\mathbb H)\leq2\beta+1.
\end{equation}
\item For each fixed $\varepsilon>0$, let $\varphi_\varepsilon$ be the $C^2$ smoothing of
$x\mapsto\operatorname{sgn}(x)|x|^\beta$ constructed in
Lemma~5.1. Define, for $0\leq t\leq T$,
\begin{equation}\label{eq:3.3}
H_t^{\beta,\varepsilon,T}:=\bigl((T+1-t)^\beta-1\bigr)\varphi_\varepsilon(S_t).
\end{equation}
Then $H^{\beta,\varepsilon,T}\in\mathcal A_T$ for every $T\geq1$.
\item The admissible family $\mathbb H^{\beta,\varepsilon}:=\{H^{\beta,\varepsilon,T}\}_{T\geq1}$ satisfies
\begin{equation}\label{eq:3.4}
C(\mathbb H^{\beta,\varepsilon})=2\beta+1.
\end{equation}
\end{enumerate}
\end{theorem}

\section{Upper bound by entropy duality}
\label{sec:4}
Fix $0<T<\infty$ and recall the process $\theta$ from \eqref{eq:2.3}.
Define the stochastic exponential
\begin{equation}\label{eq:4.1}
Z_t^T:=\exp\left\{-\int_0^t\theta_s\,dB_s-\frac12\int_0^t\theta_s^2\,ds\right\},
\qquad0\leq t\leq T.
\end{equation}
By \cite[Proposition A.1]{ref6}, $Z^T$ is a true martingale. We may therefore define a probability
measure $Q^T$ on $\mathcal F_T$ by
\begin{equation}\label{eq:4.2}
\frac{dQ^T}{d\mathbb P}=Z_T^T.
\end{equation}
Girsanov's theorem gives that
\begin{equation}\label{eq:4.3}
\widehat B_t^T:=B_t+\int_0^t\theta_s\,ds,\qquad0\leq t\leq T,
\end{equation}
is a $Q^T$-Brownian motion. Since $dS_t=\theta_t\,dt+dB_t=d\widehat B_t^T$, the process $S$ is a $Q^T$-Brownian
motion, and $Q^T$ is an equivalent martingale measure (the only one).

The weak duality inequality for exponential utility says that, for every $H\in\mathcal A_T$,
\begin{equation}\label{eq:4.4}
-\log\mathbb E^{\mathbb P}e^{-V_T^H}\leq\operatorname{Ent}(Q^T\mid\mathbb P),
\qquad \operatorname{Ent}(Q^T\mid\mathbb P):=\mathbb E^{Q^T}\log Z_T^T.
\end{equation}
This is the weak side of the minimal-entropy martingale-measure duality for exponential utility;
see \cite{ref4,ref5}. For completeness, and because only the weak inequality is needed here, we
recall the short proof. Since $d\mathbb P/dQ^T=(Z_T^T)^{-1}$,
\[
\begin{aligned}
\log\mathbb E^{\mathbb P}e^{-V_T^H}
&=\log\mathbb E^{Q^T}\exp\{-V_T^H-\log Z_T^T\}\\
&\geq-\mathbb E^{Q^T}V_T^H-\mathbb E^{Q^T}\log Z_T^T
\geq-\operatorname{Ent}(Q^T\mid\mathbb P).
\end{aligned}
\]
The last inequality uses the supermartingale property of $V^H$ under $Q^T$ and the zero initial
wealth. Hence \eqref{eq:4.4} follows. Taking the supremum over $H\in\mathcal A_T$ gives
\[
c_T\leq\operatorname{Ent}(Q^T\mid\mathbb P).
\]

It remains to estimate the entropy. Under $Q^T$, $S_t=\widehat B_t^T$ and therefore
\[
X_t=S_t-\mu t=\widehat B_t^T-\mu t.
\]
Consequently,
\[
\begin{aligned}
\mathbb E^{Q^T}\int_0^T\theta_t^2\,dt
&\leq2\mu^2T+2\alpha^2\int_0^T\mathbb E^{Q^T}|X_t|^{2\beta}\,dt\\
&\leq2\mu^2T+2\alpha^2C_\beta\int_0^T
\left(M_{2\beta}t^\beta+|\mu|^{2\beta}t^{2\beta}\right)\,dt<\infty,
\end{aligned}
\]
where $M_{2\beta}$ is the absolute moment of order $2\beta$ of a standard Gaussian random variable and
$C_\beta:=\max\{1,2^{2\beta-1}\}$. The second inequality follows from $X_t=\widehat B_t^T-\mu t$ and
\[
|x+y|^{2\beta}\leq C_\beta\left(|x|^{2\beta}+|y|^{2\beta}\right),\qquad x,y\in\mathbb R.
\]
Thus the process $\left(\int_0^t\theta_s\,d\widehat B_s^T\right)_{0\leq t\leq T}$ is a square-integrable $Q^T$-martingale with zero mean.
Since $dB_t=d\widehat B_t^T-\theta_t\,dt$,
\[
\log Z_T^T=-\int_0^T\theta_t\,d\widehat B_t^T+\frac12\int_0^T\theta_t^2\,dt,
\]
and hence
\[
\operatorname{Ent}(Q^T\mid\mathbb P)=\frac12\mathbb E^{Q^T}\int_0^T\theta_t^2\,dt.
\]
Using $(a-b)^2\leq2a^2+2b^2$, we obtain
\begin{equation}\label{eq:4.5}
\operatorname{Ent}(Q^T\mid\mathbb P)\leq\mu^2T+\alpha^2\int_0^T\mathbb E^{Q^T}|X_t|^{2\beta}\,dt.
\end{equation}
Set $p:=2\beta\in(0,2)$. Since $|x+y|^p\leq C_p(|x|^p+|y|^p)$ and
$\mathbb E^{Q^T}|\widehat B_t^T|^p=C_pt^{p/2}=C_pt^\beta$, we have
\[
\mathbb E^{Q^T}|X_t|^{2\beta}\leq C_{\beta,\mu}(t^\beta+t^{2\beta}).
\]
Therefore
\begin{equation}\label{eq:4.6}
\int_0^T\mathbb E^{Q^T}|X_t|^{2\beta}\,dt\leq C_{\beta,\mu}(T^{\beta+1}+T^{2\beta+1}).
\end{equation}
Combining \eqref{eq:4.5} and \eqref{eq:4.6}, and using $T=O(T^{2\beta+1})$ and $T^{\beta+1}=O(T^{2\beta+1})$ for $\beta>0$, yields
\begin{equation}\label{eq:4.7}
c_T\leq\operatorname{Ent}(Q^T\mid\mathbb P)\leq C_{\alpha,\beta,\mu}T^{2\beta+1},\qquad T\geq1.
\end{equation}
It follows directly from \eqref{eq:4.4} and \eqref{eq:4.7} that no admissible family can have order larger than
$2\beta+1$.

\section{Construction of the smoothed power}
For $0<\beta<1$, the function $x\mapsto\operatorname{sgn}(x)|x|^\beta$
is not $C^2$ at the origin. We therefore replace it on a fixed interval $[-\varepsilon,\varepsilon]$ and record the
properties needed later.

\begin{lemma}[Smoothed power and primitive]\label{lem:5.1}
Fix $0<\beta<1$ and $\varepsilon>0$. Define the odd
function $\varphi_\varepsilon$ by
\begin{equation}\label{eq:5.1}
\varphi_\varepsilon(x)=Ax+Bx^3+Cx^5\quad (|x|\leq\varepsilon),
\qquad
\varphi_\varepsilon(x)=\operatorname{sgn}(x)|x|^\beta\quad (|x|\geq\varepsilon),
\end{equation}
where
\begin{equation}\label{eq:5.2}
A=\frac{15-8\beta+\beta^2}{8}\varepsilon^{\beta-1},
\qquad
B=\frac{(\beta-1)(5-\beta)}{4}\varepsilon^{\beta-3},
\qquad
C=\frac{(\beta-1)(\beta-3)}{8}\varepsilon^{\beta-5}.
\end{equation}
Then $\varphi_\varepsilon\in C^2(\mathbb{R})$. If
\begin{equation}\label{eq:5.3}
F_\varepsilon(x):=\int_0^x\varphi_\varepsilon(u)\,du,
\end{equation}
then $F_\varepsilon\in C^3(\mathbb{R})$, $F_\varepsilon$ is even, and
\begin{equation}\label{eq:5.4}
F_\varepsilon(x)=\frac{A}{2}x^2+\frac{B}{4}x^4+\frac{C}{6}x^6,
\qquad |x|\leq\varepsilon,
\end{equation}
while, for $|x|\geq\varepsilon$,
\begin{equation}\label{eq:5.5}
F_\varepsilon(x)=\frac{|x|^{\beta+1}}{\beta+1}+D_\varepsilon,
\qquad
D_\varepsilon=\varepsilon^{\beta+1}\frac{(\beta-5)(\beta-3)(\beta-1)}{48(\beta+1)}<0.
\end{equation}
Moreover, there are finite constants $K_\varepsilon$ and $M_\varepsilon$ such that, for all $x\in\mathbb{R}$,
\begin{equation}\label{eq:5.6}
F_\varepsilon(x)\geq0,
\qquad
F_\varepsilon(x)\geq\frac{|x|^{\beta+1}}{\beta+1}-K_\varepsilon,
\qquad
\|F_\varepsilon''\|_\infty\leq M_\varepsilon.
\end{equation}
\end{lemma}

\begin{proof}
The coefficients in \eqref{eq:5.2} are obtained by imposing $C^2$ matching at $x=\varepsilon$:
\begin{align}
A\varepsilon+B\varepsilon^3+C\varepsilon^5&=\varepsilon^\beta,\label{eq:5.7}\\
A+3B\varepsilon^2+5C\varepsilon^4&=\beta\varepsilon^{\beta-1},\label{eq:5.8}\\
6B\varepsilon+20C\varepsilon^3&=\beta(\beta-1)\varepsilon^{\beta-2}.\label{eq:5.9}
\end{align}
Oddness gives the same matching at $x=-\varepsilon$, so $\varphi_\varepsilon\in C^2(\mathbb{R})$. The formula \eqref{eq:5.4} follows by
integration on $[-\varepsilon,\varepsilon]$. For $|x|\geq\varepsilon$,
\[
F_\varepsilon(x)=F_\varepsilon(\varepsilon)+\int_\varepsilon^{|x|}u^\beta\,du
=\frac{|x|^{\beta+1}}{\beta+1}+D_\varepsilon,
\]
with $D_\varepsilon$ as in \eqref{eq:5.5}. This additive constant is important: the primitive does not equal
$|x|^{\beta+1}/(\beta+1)$ outside the smoothing interval unless the inner spline is chosen differently.

It remains to prove nonnegativity and boundedness of the second derivative. For $0\leq x\leq\varepsilon$,
set $y=(x/\varepsilon)^2$. Then
\[
\varphi_\varepsilon(x)=x\varepsilon^{\beta-1}q(y),
\]
where
\[
q(y)=\frac{15-8\beta+\beta^2}{8}
+\frac{(\beta-1)(5-\beta)}{4}y
+\frac{(\beta-1)(\beta-3)}{8}y^2.
\]
We have $q(1)=1$, and for $0\leq y\leq1$,
\begin{equation}\label{eq:5.10}
q'(y)=\frac{(\beta-1)(5-\beta)}{4}
+\frac{(\beta-1)(\beta-3)}{4}y
\leq\frac{\beta-1}{2}<0.
\end{equation}
Hence $q(y)\geq q(1)=1$, so $\varphi_\varepsilon(x)\geq0$ for $x\geq0$. Since $\varphi_\varepsilon$ is odd, $F_\varepsilon$ is even and nonnegative.
The lower bound in \eqref{eq:5.6} follows from \eqref{eq:5.5} outside $[-\varepsilon,\varepsilon]$ and from compactness on $[-\varepsilon,\varepsilon]$.
Finally, $F_\varepsilon''=\varphi_\varepsilon'$ is bounded on $[-\varepsilon,\varepsilon]$ by the polynomial formula, while for $|x|\geq\varepsilon$,
\begin{equation}\label{eq:5.11}
|F_\varepsilon''(x)|=\beta|x|^{\beta-1}\leq\beta\varepsilon^{\beta-1}.
\end{equation}
Thus $\|F_\varepsilon''\|_\infty<\infty$.
\end{proof}

\section{It\^o identity and admissibility of the candidate}
The next lemma is the analytic core
of the construction.

\begin{lemma}[Terminal wealth decomposition and admissibility]\label{lem:6.1}
Fix $\varepsilon>0$ and $T>0$, and
write $\varphi=\varphi_\varepsilon$ and $F=F_\varepsilon$. Define
\begin{equation}\label{eq:6.1}
H_t^{\beta,\varepsilon,T}:=w_T(t)\varphi(S_t),
\qquad 0\leq t\leq T,
\end{equation}
where
\begin{equation}\label{eq:6.2}
w_T(t):=(T+1-t)^\beta-1.
\end{equation}
Then $H^{\beta,\varepsilon,T}\in\mathcal{A}_T$. Its terminal wealth admits the decomposition
\begin{equation}\label{eq:6.3}
V_T^{\beta,\varepsilon}:=(H^{\beta,\varepsilon,T}\cdot S)_T=I_1(T)-I_2(T),
\end{equation}
where
\begin{equation}\label{eq:6.4}
I_1(T):=\beta\int_0^T(T+1-t)^{\beta-1}F(S_t)\,dt
\end{equation}
and
\begin{equation}\label{eq:6.5}
I_2(T):=\frac{1}{2}\int_0^Tw_T(t)F''(S_t)\,dt.
\end{equation}
Moreover,
\begin{equation}\label{eq:6.6}
|I_2(T)|\leq B_\varepsilon T^{\beta+1},
\qquad
B_\varepsilon:=\frac{1}{2(\beta+1)}\|F''\|_\infty.
\end{equation}
\end{lemma}

\begin{proof}
The process $H^{\beta,\varepsilon,T}$ is predictable because it is continuous and adapted, and it is
$S$-integrable on finite horizons because $\varphi(S)$ is locally bounded. Let
\[
u(t,x):=w_T(t)F(x),\qquad 0\leq t\leq T.
\]
Then $u\in C^{1,2}([0,T]\times\mathbb{R})$, $w_T(T)=0$, and $F(S_0)=F(0)=0$. It\^o's formula on $[0,T]$ gives
\[
du(t,S_t)=-\beta(T+1-t)^{\beta-1}F(S_t)\,dt+w_T(t)\varphi(S_t)\,dS_t
+\tfrac{1}{2}w_T(t)F''(S_t)\,dt.
\]
The boundary terms vanish, and hence
\begin{equation}\label{eq:6.7}
0=-\beta\int_0^T(T+1-t)^{\beta-1}F(S_t)\,dt
+\int_0^Tw_T(t)\varphi(S_t)\,dS_t
+\frac{1}{2}\int_0^Tw_T(t)F''(S_t)\,dt.
\end{equation}
This gives \eqref{eq:6.3}--\eqref{eq:6.5}. The bound \eqref{eq:6.6} follows from Lemma~\ref{lem:5.1} and the subadditivity inequality $(1+s)^\beta\leq1+s^\beta$, valid for $s\geq0$ and $0<\beta<1$:
\begin{align*}
|I_2(T)|&\leq\frac{1}{2}\|F''\|_\infty\int_0^T\bigl((T+1-t)^\beta-1\bigr)\,dt\\
&=\frac{1}{2}\|F''\|_\infty\int_0^T\bigl((1+s)^\beta-1\bigr)\,ds
\leq\frac{1}{2(\beta+1)}\|F''\|_\infty T^{\beta+1}.
\end{align*}
It remains to show admissibility. The same calculation on $[0,r]$, $0\leq r\leq T$, gives
\begin{align*}
(H^{\beta,\varepsilon,T}\cdot S)_r
&=w_T(r)F(S_r)+\beta\int_0^r(T+1-t)^{\beta-1}F(S_t)\,dt\\
&\quad-\frac{1}{2}\int_0^rw_T(t)F''(S_t)\,dt.
\end{align*}
The first two terms are nonnegative by Lemma~\ref{lem:5.1}, and the last term is bounded below by
$-B_\varepsilon T^{\beta+1}$. Hence
\begin{equation}\label{eq:6.8}
(H^{\beta,\varepsilon,T}\cdot S)_r\geq-B_\varepsilon T^{\beta+1},
\qquad 0\leq r\leq T.\qedhere
\end{equation}
\end{proof}

\section{Exponential moment for the mean-reverting factor}
The following lemma is a standard Foster--Lyapunov-type consequence of a generator drift inequality.

\begin{lemma}[Exponential moment of the occupation functional]\label{lem:7.1}
Let
\begin{equation}\label{eq:7.1}
\Xi_t:=\int_0^t|X_s|^{2\beta}\,ds.
\end{equation}
There are constants $\delta_0,c_0,C_0>0$, depending only on $\alpha$ and $\beta$, such that, for every $t\geq0$ and
every tail level $r>0$,
\begin{equation}\label{eq:7.2}
\mathbb{E}\exp\{\delta_0\Xi_t\}\leq C_0e^{c_0t},
\qquad
\mathbb{P}(\Xi_t\geq r)\leq C_0\exp\{c_0t-\delta_0r\}.
\end{equation}
\end{lemma}

\begin{proof}
Let
\[
Lf(x)=-\alpha\operatorname{sgn}(x)|x|^\beta f'(x)+\frac{1}{2}f''(x)
\]
be the generator of the diffusion process $X$. Choose an even $C^2$ function $V\geq1$ such that
$V(x)=1+a|x|^{\beta+1}$ for $|x|\geq2$, with the function smoothly modified on $[-2,2]$. This smoothing
avoids the singularity of the second derivative at the origin. For $|x|\geq2$,
\[
V'(x)=a(\beta+1)\operatorname{sgn}(x)|x|^\beta,
\qquad
V''(x)=a\beta(\beta+1)|x|^{\beta-1},
\]
and therefore
\[
LV(x)=-\alpha a(\beta+1)|x|^{2\beta}+\frac{1}{2}a\beta(\beta+1)|x|^{\beta-1}.
\]
Since $2\beta>\beta-1$, and compact sets can be absorbed into constants, there are constants
$c_1,c_2,c_3,c_4>0$ such that, for all $x\in\mathbb{R}$,
\[
LV(x)\leq-c_1|x|^{2\beta}+c_2,
\qquad
(V'(x))^2\leq c_3|x|^{2\beta}+c_4.
\]
For $W_\lambda=e^{\lambda V}$,
\[
\frac{LW_\lambda}{W_\lambda}=\lambda LV+\frac{\lambda^2}{2}(V')^2
\leq-\left(\lambda c_1-\frac{\lambda^2c_3}{2}\right)|x|^{2\beta}
+\lambda c_2+\frac{\lambda^2c_4}{2}.
\]
Choose $\lambda>0$ so small that
\[
\delta_0:=\lambda c_1-\frac{\lambda^2c_3}{2}>0,
\qquad
c_0:=\lambda c_2+\frac{\lambda^2c_4}{2}.
\]
Then
\begin{equation}\label{eq:7.3}
(L+\delta_0|x|^{2\beta})W_\lambda(x)\leq c_0W_\lambda(x),
\qquad x\in\mathbb{R}.
\end{equation}
Let $\tau_n:=\inf\{t\geq0:|X_t|\geq n\}$. Since the SDE has a nonexplosive strong solution, $\tau_n\uparrow\infty$
a.s. Applying It\^o's formula to the stopped process
\[
\exp\left\{\delta_0\int_0^{t\wedge\tau_n}|X_s|^{2\beta}\,ds-c_0(t\wedge\tau_n)\right\}
W_\lambda(X_{t\wedge\tau_n})
\]
is legitimate, because all coefficients are bounded up to $\tau_n$. By \eqref{eq:7.3} its drift is nonpositive;
hence this stopped process is a nonnegative supermartingale. Since $W_\lambda\geq1$,
\[
\mathbb{E}\exp\left\{\delta_0\int_0^{t\wedge\tau_n}|X_s|^{2\beta}\,ds\right\}
\leq e^{c_0t}W_\lambda(0).
\]
Letting $n\to\infty$ and using Fatou's lemma gives $\mathbb{E}e^{\delta_0\Xi_t}\leq W_\lambda(0)e^{c_0t}$.

Thus \eqref{eq:7.2} holds with $C_0:=W_\lambda(0)$, and the tail estimate follows immediately from
Markov's inequality.
\end{proof}

\section{\texorpdfstring{Lower bound for $H^{\beta,\varepsilon,T}$ when $1/2\le\beta<1$}{Lower bound for H (β, ε, T) when 1/2 ≤ β < 1}}
\label{sec:8}

We first give a shorter argument available when $2\beta\ge1$, because it is close to the original proof of \cite{ref6}. Fix $\eta\in(0,1/8)$, say $\eta=1/16$, and define
\[
A_T:=\left\{\left|\int_0^{T/2}X_t\,dt\right|\le\eta|\mu|T^2\right\}.
\]
On $A_T$,
\[
\left|\int_0^{T/2}S_t\,dt\right|
=\left|\frac{\mu T^2}{8}+\int_0^{T/2}X_t\,dt\right|
\ge(1/8-\eta)|\mu|T^2\ge\frac{|\mu|}{16}T^2.
\]
Jensen's inequality for the convex function $x\mapsto|x|^{\beta+1}$ gives
\begin{equation}
\label{eq:8.1}
\int_0^{T/2}|S_t|^{\beta+1}\,dt
\ge\frac{T}{2}\left|\frac{2}{T}\int_0^{T/2}S_t\,dt\right|^{\beta+1}
\ge K_{\beta,\mu}T^{\beta+2},
\end{equation}
where
\[
K_{\beta,\mu}:=\frac{1}{2}\left(\frac{|\mu|}{8}\right)^{\beta+1}.
\]
Since $(T+1-t)^{\beta-1}\ge(T+1)^{\beta-1}\ge2^{\beta-1}T^{\beta-1}$ for $0\le t\le T/2$ and $T\ge1$,
\[
\begin{aligned}
I_1(T)&\ge\beta2^{\beta-1}T^{\beta-1}\int_0^{T/2}F(S_t)\,dt\\
&\ge\frac{\beta}{\beta+1}\,2^{\beta-1}T^{\beta-1}\int_0^{T/2}|S_t|^{\beta+1}\,dt-\beta2^{\beta-2}K_\varepsilon T^\beta\\
&\ge\frac{\beta}{\beta+1}\,2^{\beta-1}K_{\beta,\mu}T^{2\beta+1}-\beta2^{\beta-2}K_\varepsilon T^\beta
\end{aligned}
\]
on $A_T$. Together with \eqref{eq:6.6}, this implies that for some $a>0$ and all sufficiently large $T$,
\begin{equation}
\label{eq:8.2}
V_T^{\beta,\varepsilon}=I_1(T)-I_2(T)\ge aT^{2\beta+1}\qquad\text{on }A_T.
\end{equation}
Consequently, for all sufficiently large $T$
\begin{equation}
\label{eq:8.3}
\mathbb{E}\bigl[e^{-V_T^{\beta,\varepsilon}}\mathbf{1}_{A_T}\bigr]\le e^{-aT^{2\beta+1}}.
\end{equation}

It remains to bound the probability of $A_T^c:=\Omega\setminus A_T$. If $\beta\ge1/2$, then $x\mapsto|x|^{2\beta}$ is convex, so Jensen's inequality applied to the average over $[0,T/2]$ yields
\[
\left|\frac{2}{T}\int_0^{T/2}X_t\,dt\right|^{2\beta}
\le\frac{2}{T}\int_0^{T/2}|X_t|^{2\beta}\,dt.
\]
On $A_T^c$, the left-hand side is larger than $(2\eta|\mu|T)^{2\beta}$. Therefore
\begin{equation}
\label{eq:8.4}
A_T^c\subseteq\left\{\int_0^{T/2}|X_t|^{2\beta}\,dt\ge c_{\mu,\eta}T^{2\beta+1}\right\},
\qquad c_{\mu,\eta}:=\frac{1}{2}(2\eta|\mu|)^{2\beta}>0.
\end{equation}
By Lemma 7.1,
\[
\mathbb{P}(A_T^c)\le C_0\exp\left\{\frac{c_0}{2}T-\delta_0c_{\mu,\eta}T^{2\beta+1}\right\}\le Ce^{-bT^{2\beta+1}}
\]
for all sufficiently large $T$, because $2\beta+1>1$. Since $F\ge0$, we have $I_1(T)\ge0$, and \eqref{eq:6.6} gives $e^{-V_T^{\beta,\varepsilon}}\le e^{B_\varepsilon T^{\beta+1}}$. Hence
\begin{equation}
\label{eq:8.5}
\mathbb{E}\bigl[e^{-V_T^{\beta,\varepsilon}}\mathbf{1}_{A_T^c}\bigr]
\le e^{B_\varepsilon T^{\beta+1}}\mathbb{P}(A_T^c)\le Ce^{-b'T^{2\beta+1}}
\end{equation}
for all sufficiently large $T$. Combining the good-event estimate \eqref{eq:8.3} with the bad-event estimate \eqref{eq:8.5} gives
\begin{equation}
\label{eq:8.6}
\mathbb{E}e^{-V_T^{\beta,\varepsilon}}\le Ce^{-bT^{2\beta+1}},\qquad 1/2\le\beta<1,
\end{equation}
for suitable constants $C,b>0$.

\section{\texorpdfstring{Lower bound for $H^{\beta,\varepsilon,T}$ when $0<\beta<1/2$}{Lower bound for H (β, ε, T) when 0 < β < 1/2}}
\label{sec:9}

When $0<\beta<1/2$, the function $x\mapsto|x|^{2\beta}$ is not convex, so the preceding estimate for $A_T^c$, based on Jensen's inequality, is unavailable. We use instead an event where the averaged size of $S$ is large enough. This argument is also valid for $\beta\ge1/2$, but it is needed only in the present non-convex regime. Let
\[
q:=\frac{|\mu|}{8},\qquad\kappa:=\frac{q^{\beta+1}}{16},\qquad
G_T:=\left\{\int_0^{T/2}|S_t|^{\beta+1}\,dt\ge\kappa T^{\beta+2}\right\}.
\]
Assume $T\ge1$, then $(T+1-t)^{\beta-1}\ge(2T)^{\beta-1}$ holds hence on $G_T$ the lower bound \eqref{eq:5.6} gives
\[
\begin{aligned}
I_1(T)&\ge\beta2^{\beta-1}T^{\beta-1}\int_0^{T/2}F(S_t)\,dt\\
&\ge\frac{\beta}{\beta+1}2^{\beta-1}T^{\beta-1}\int_0^{T/2}|S_t|^{\beta+1}\,dt-\beta2^{\beta-2}K_\varepsilon T^\beta\\
&\ge\frac{\beta}{\beta+1}2^{\beta-1}\kappa T^{2\beta+1}-C_\varepsilon T^\beta.
\end{aligned}
\]
Combining this with \eqref{eq:6.6} and using $\beta+1<2\beta+1$, we obtain constants $a_1,T_1>0$ such that
\[
V_T^{\beta,\varepsilon}\ge a_1T^{2\beta+1}\qquad\text{on }G_T,\qquad T\ge T_1.
\]
Consequently,
\begin{equation}
\label{eq:9.1}
\mathbb{E}\bigl[e^{-V_T^{\beta,\varepsilon}}\mathbf{1}_{G_T}\bigr]\le e^{-a_1T^{2\beta+1}},\qquad T\ge T_1.
\end{equation}

We now show that $G_T^c:=\Omega\setminus G_T$ has exponentially small probability at the right scale. Set
\[
J_T:=[T/4,T/2],\qquad D_T:=\{t\in J_T:|S_t|>qT\}.
\]
We will denote by $\ell(A)$ the Lebesgue measure of a Borel set $A\subset\mathbb{R}$. On $G_T^c$,
\[
\ell(D_T)(qT)^{\beta+1}\le\int_{J_T}|S_t|^{\beta+1}\,dt<\kappa T^{\beta+2},
\]
so $\ell(D_T)<\kappa q^{-(\beta+1)}T=T/16$. Since $\ell(J_T)=T/4$, this implies $\ell(J_T\setminus D_T)>3T/16$. For $t\in J_T\setminus D_T$,
\[
|X_t|=|S_t-\mu t|\ge|\mu|t-|S_t|\ge\frac{|\mu|T}{4}-qT=qT.
\]
Consequently,
\begin{equation}
\label{eq:9.2}
G_T^c\subseteq\left\{\int_0^{T/2}|X_t|^{2\beta}\,dt\ge\gamma_{\mu,\beta}T^{2\beta+1}\right\},
\qquad\gamma_{\mu,\beta}:=\frac{3}{16}\left(\frac{|\mu|}{8}\right)^{2\beta}>0.
\end{equation}
Lemma 7.1 gives
\[
\mathbb{P}(G_T^c)\le C_0\exp\left\{\frac{c_0}{2}T-\delta_0\gamma_{\mu,\beta}T^{2\beta+1}\right\}\le Ce^{-bT^{2\beta+1}}
\]
for all sufficiently large $T$. Since $F\ge0$, $I_1(T)\ge0$, and hence $e^{-V_T^{\beta,\varepsilon}}\le e^{B_\varepsilon T^{\beta+1}}$, the contribution of $G_T^c$ satisfies
\begin{equation}
\label{eq:9.3}
\mathbb{E}\bigl[e^{-V_T^{\beta,\varepsilon}}\mathbf{1}_{G_T^c}\bigr]
\le e^{B_\varepsilon T^{\beta+1}}\mathbb{P}(G_T^c)\le Ce^{-b'T^{2\beta+1}}
\end{equation}
for all sufficiently large $T$. Combining the good-event estimate \eqref{eq:9.1} with the preceding bad-event estimate \eqref{eq:9.3} gives
\begin{equation}
\label{eq:9.4}
\mathbb{E}e^{-V_T^{\beta,\varepsilon}}\le Ce^{-bT^{2\beta+1}},\qquad 0<\beta<1/2.
\end{equation}
Thus the same lower bound holds in the non-convex regime, too.

\section{Completion of the proof}
\label{sec:10}

Equations \eqref{eq:8.6} and \eqref{eq:9.4} show that, for every $0<\beta<1$, every fixed $\varepsilon>0$, and all sufficiently large $T$, there are constants $C,b>0$ such that
\begin{equation}
\label{eq:10.1}
\mathbb{E}\exp\{-V_T^{\beta,\varepsilon}\}\le Ce^{-bT^{2\beta+1}}.
\end{equation}
Therefore
\[
\liminf_{T\to\infty}\frac{-\log\mathbb{E}\exp\{-V_T^{\beta,\varepsilon}\}}{T^{2\beta+1}}\ge b>0,
\]
and hence $C(\mathbb{H}^{\beta,\varepsilon})\ge2\beta+1$. The upper bound from section 4 gives $C(\mathbb{H})\le2\beta+1$ for every admissible family $\mathbb{H}$. Since the family $\mathbb{H}^{\beta,\varepsilon}$ is admissible by Section 6, we conclude that $C(\mathbb{H}^{\beta,\varepsilon})=2\beta+1$. This proves Theorem 3.1.

\end{document}